\documentclass[english]{article}
\usepackage[T1]{fontenc}
\usepackage[latin9]{inputenc}
\usepackage{geometry}
\usepackage{color}
\usepackage{babel}
\usepackage{amsmath}
\usepackage{amsthm}
\usepackage{amssymb}
\usepackage{setspace}
\usepackage[unicode=true]
 {hyperref}

\makeatletter
  \theoremstyle{plain}
  \newtheorem{conjecture}{\protect\conjecturename}
  \theoremstyle{definition}
  \newtheorem{defn}{\protect\definitionname}
  \theoremstyle{remark}
  \newtheorem{rem}{\protect\remarkname}
  \theoremstyle{plain}
  \newtheorem{lem}{\protect\lemmaname}
\theoremstyle{plain}
\newtheorem{thm}{\protect\theoremname}
  \theoremstyle{plain}
  \newtheorem{prop}{\protect\propositionname}

\makeatother

  \providecommand{\conjecturename}{Conjecture}
  \providecommand{\definitionname}{Definition}
  \providecommand{\lemmaname}{Lemma}
  \providecommand{\propositionname}{Proposition}
  \providecommand{\remarkname}{Remark}
\providecommand{\theoremname}{Theorem}

\begin{document}
\begin{doublespace}
\begin{center}
\textbf{\large{}Fighting Uncertainty with Uncertainty: A Baby Step}
\par\end{center}{\large \par}

\begin{center}
\textbf{Ravi Kashyap }
\par\end{center}

\begin{center}
\textbf{City University of Hong Kong }
\par\end{center}

\begin{center}
\today
\par\end{center}

\begin{center}
\par\end{center}

\begin{center}
Risk and Uncertainty
\par\end{center}

\begin{center}
\textbf{\textcolor{blue}{\href{https://doi.org/10.4236/tel.2017.75097}{Edited Version: Kashyap, R. (2017). Fighting Uncertainty with Uncertainty: A Baby Step. Theoretical Economics Letters, 7(5), 1431-1452.}}}
\par\end{center}

\begin{center}
\tableofcontents{}
\par\end{center}
\end{doublespace}
\begin{doublespace}

\section{Abstract }
\end{doublespace}

\begin{doublespace}
We can overcome uncertainty with uncertainty. Using randomness in
our choices and in what we control, and hence in the decision making
process, could potentially offset the uncertainty inherent in the
environment and yield better outcomes. The example we develop in greater
detail is the news-vendor inventory management problem with demand
uncertainty. We briefly discuss areas, where such an approach might
be helpful, with the common prescription, ``Don't Simply Optimize,
Also Randomize; perhaps best described by the term - Randoptimization''.
\end{doublespace}
\begin{enumerate}
\begin{doublespace}
\item News-vendor Inventory Management
\item School Admissions
\item Journal Submissions
\item Job Candidate Selection
\item Stock Picking
\item Monetary Policy
\end{doublespace}
\end{enumerate}
\begin{doublespace}
This methodology is suitable for the social sciences since the primary
source of uncertainty are the members of the system themselves and
presently, no methods are known to fully determine the outcomes in
such an environment, which perhaps would require being able to read
the minds of everyone involved and to anticipate their actions continuously.
Admittedly, we are not qualified to recommend whether such an approach
is conducive for the natural sciences, unless perhaps, bounds can
be established on the levels of uncertainty in a system and it is
shown conclusively that a better understanding of the system and hence
improved decision making will not alter the outcomes. 
\end{doublespace}
\begin{doublespace}

\section{Introduction }
\end{doublespace}
\begin{doublespace}

\subsection{True Comparison Theory or Lack Thereof}
\end{doublespace}

\begin{doublespace}
(Keeney \cite{key-1}) presents an overview of decision analysis:
what it is, what it can and cannot do, and how to do it, including
a summary of the axioms, while maintaining that ``complexity cannot
be avoided in making decisions''. He defines decision analysis informally
as \textquotedbl{}a formalization of common sense for decision problems
which are too complex for informal use of common sense\textquotedbl{};
and more technically as \textquotedbl{}a philosophy, articulated by
a set of logical axioms, and a methodology and collection of systematic
procedures, based upon those axioms, for responsibly analyzing the
complexities inherent in decision problems.\textquotedbl{} (Pratt,
Raiffa \& Schlaifer \cite{key-1-1}) is a detailed discussion of the
axioms of decision analysis; (Raiffa \cite{key-2-2}) is an interesting
personal account. The similarities of this approach to subjective
or personal evaluation of probabilities to formulate decisions under
uncertainty (Ramsey \cite{key-2-3}; Savage \cite{key-2-4}) are to
be noted. 

If we admit that common sense cannot deal with complexity, especially
when subjective decisions from the participants themselves are creating
complexity; we need to be open to the possibility that perhaps, extremely
careful analysis leading to improved decision making might merely
be an illusion. A hall mark of the social sciences is the lack of
objectivity. Here we assert that objectivity is with respect to comparisons
done by different participants and that a comparison is a precursor
to a decision (Kashyap \cite{key-2-5}).
\end{doublespace}
\begin{conjecture}
\begin{doublespace}
\textbf{Despite the several advances in the social sciences,} \textbf{we
have yet to discover an objective measuring stick for comparison,
a so called, True Comparison Theory, which can be an aid for arriving
at objective decisions}. 
\end{doublespace}
\end{conjecture}
\begin{doublespace}
The search for such a theory could again be compared, to the medieval
alchemists\textquoteright{} obsession with turning everything into
gold (Kashyap \cite{key-2-6}). For our present purposes, the lack
of such an objective measure means that the difference in comparisons,
as assessed by different participants, can effect different decisions
under the same set of circumstances. Hence, despite all the uncertainty
in the social sciences, the one thing we can be almost certain about
is the subjectivity in all decision making.

This lack of an objective measure for comparisons, makes people react
at varying degrees and at varying speeds, as they make their subjective
decisions. A decision gives rise to an action and subjectivity in
the comparison means differing decisions and hence unpredictable actions.
This inability to make consistent predictions in the social sciences
explains the growing trend towards comprehending better and deciphering
the decision process and the subsequent actions, by collecting more
information across the entire cycle of comparisons, decisions and
actions. Another feature of the social sciences is that the actions
of participants affects the state of the system, effecting a state
transfer which perpetuates another merry-go-round of comparisons,
decisions and actions from the participants involved. This means,
more the participants, more the changes to the system, more the actions
and more the information that is generated to be gathered.
\end{doublespace}
\begin{doublespace}

\subsection{The Uncertainty Principle of the Social Sciences}
\end{doublespace}

\begin{doublespace}
(Lawson \cite{key-2-7}) argues that the Keynesian view on uncertainty
(that it is generally impossible, even in probabilistic terms, to
evaluate the future outcomes of all possible current actions: Keynes
\cite{key-2-8}; \cite{key-3}; \cite{key-4}), far from being innocuous
or destructive of economic analysis in general, can give rise to research
programs incorporating, amongst other things, a view of rational behavior
under uncertainty, which could be potentially fruitful. (Knight \cite{key-4-1})
in his seminal work, ``Risk, Uncertainty and Profit'', makes a distinction
between Uncertainty and Risk, which can be immensely enlightening
in our efforts to comprehend uncertainty.

First, we note a difference between the Keynesian and Knightian views
of uncertainty (Hoogduin \cite{key-5}): Knight mainly focuses on
the distinction between numerically measurable and not numerically
measurable probabilities; Keynes stresses the slight amount of knowledge
on which probabilities often have to be based. There are many definitions
of uncertainty; for simplicity and to suit our present purpose, we
consider the following definition based on the Keynesian and Knightian
views:
\end{doublespace}
\begin{defn}
\begin{doublespace}
\textbf{\textit{Uncertainty is the lack of certainty; a situation
of having limited knowledge where it is impossible to exactly describe
the existing state, a future outcome, or even be aware of all the
possible outcomes. Risk, means, in some cases a quantity susceptible
of measurement, while at other times it is something distinctly not
of this character. Uncertainty then becomes risk that is immeasurable
and not possible to calculate}}\footnote{\begin{doublespace}
With this definition of risk and uncertainty, our paper must perhaps
be titled: Fighting Uncertainty with Risk. Overlooking the possibility
that the present title is more catchy, we note that our attempts at
controlled or measurable uncertainty might lead to unforeseen consequences,
which might not be easily measured. In general, even when we believe
that we have a good idea of the risks, we might be over estimating
our ability to estimate the risks (the effects of our actions on any
social system and the outcomes therein), which is precisely the definition
of uncertainty we have used.

To provide readers a taste of the many flavors of commonly used definitions,
we note that: Risk is also the potential of gaining or losing something
of value (Kungwani \cite{key-42}). Values (such as physical health,
social status, emotional well-being, or financial wealth) can be gained
or lost when taking risk, resulting from a given action or inaction,
foreseen or unforeseen. Risk can also be defined as the intentional
interaction with uncertainty. Risk can be seen as a state of uncertainty
where some possible outcomes have an undesired effect or significant
loss. Measurement of risk involves coming up with a set of measured
uncertainties where some possible outcomes are losses, and factoring
in the magnitudes of those losses. Measurement of uncertainty involves
an effort to assess as best as feasible, a set of possible states
or outcomes where probabilities are assigned to each possible state
or outcome.
\end{doublespace}
}\textbf{\textit{.}}
\end{doublespace}
\end{defn}
\begin{doublespace}
(McManus and Hastings \cite{key-6}) clarify the wide range of uncertainties
that affect complex engineering systems and present a framework to
understand the risks (and opportunities) they create and the strategies
system designers can use to mitigate or take advantage of them. These
viewpoints hold many lessons for policy designers in the social sciences
and could be instructive for researchers looking at ways to understand
and contend with complex systems, keeping in mind the caveats of dynamic
social systems. (Kashyap \cite{key-2-6}) discusses the uncertainty
principle of the social sciences and the use of a feedback loop to
overcome the phenomenon, wherein participants observe the results
and change their actions that could potentially lead to unexpected
and different consequences.
\end{doublespace}
\begin{conjecture}
\begin{doublespace}
\textbf{The Uncertainty Principle of the Social Sciences can be stated
as, \textquotedblleft Any generalization in the social sciences cannot
be both popular and continue to yield predictions, or in other words,
the more popular a particular generalization, the less accurate will
be the predictions it yields\textquotedblright .} 
\end{doublespace}
\end{conjecture}
\begin{doublespace}
This is because as soon as any generalization and its set of conditions
becomes common knowledge, the entry of many participants shifts the
equilibrium or the dynamics, such that the generalization no longer
applies to the known set of conditions. An observation is likely to
be more popular when there are more people comprising that system;
and it is important to try and explicitly understand, where possible,
how predictions can go awry. Every social system then operates under
the overarching reach of this principle. 

This varying behavior of participants in a social system will give
rise to unintended consequences (Kashyap \cite{key-7}; \cite{key-8})
and as long as participants are free to observe the results and modify
their actions, this effect will persist\footnote{\begin{doublespace}
Sergei Bubka {[}\href{http://en.wikipedia.org/wiki/Sergey_Bubka}{Sergey Bubka: http://en.wikipedia.org/wiki/Sergey\_{}Bubka}{]}
is our Icon of Uncertainty. As a refresher for the younger generation,
he broke the pole vault world record 35 times. We can think of regulatory
change or the utilization of newer methods and techniques as raising
the bar. Each time the bar is raised, the spirit of Sergei Bubka,
in all of us, will find a way over it.
\end{doublespace}
}. (Simon \cite{key-9}) points out that any attempt to seek properties
common to many sorts of complex systems (physical, biological or social),
would lead to a theory of hierarchy, since a large proportion of complex
systems observed in nature exhibit hierarchic structure.

(Kashyap \cite{key-10}; \cite{key-10-1}) consider ways to reduce
the complexity of social systems, which could be one way to mitigate
the effect of unintended outcomes. While attempts at designing less
complex systems are worthy endeavors, reduced complexity might be
hard to accomplish in certain instances and despite successfully reducing
complexity, alternate techniques at dealing with uncertainty are commendable
complementary pursuits. It might be possible to observe historical
trends (or other attributes) and make comparisons across fewer number
of entities; though in large systems where there are numerous components
or contributing elements, this can be a daunting task and constructing
measures across groups or aggregations of smaller systems need to
be explored (Kashyap \cite{key-2-5}; \cite{key-10-2}).
\end{doublespace}
\begin{doublespace}

\subsection{Randoptimization}
\end{doublespace}

\begin{doublespace}
(Kleywegt \& Shapiro \cite{key-10-3}) is a detailed account of decision
making under uncertainty and how decision problems are often formulated
as optimization problems and solved using stochastic optimization
techniques. In this present paper, we take a step further and look
at ways in which introducing randomness in our choices and in what
we control, and hence in the decision making process, could potentially
offset the uncertainty inherent in the environment and yield better
outcomes. 

In short, we try to overcome uncertainty with uncertainty. Such an
approach, while seemingly absurd, is the ideal medicine for the even
greater absurdity in the decision making process that has become prevalent
in today's society. It takes care of the issues that crop up due to
the limitations in our understanding of complex systems, and the widely
acknowledged problem that most of our measurements are highly prone
to errors. 
\end{doublespace}
\begin{rem}
\begin{doublespace}
\textbf{\textit{Necessity is the mother of all creation (invention
or innovation), but the often forgotten father, is frustration.}}
\end{doublespace}
\end{rem}
\begin{doublespace}
Pondering on the sources of uncertainty and the tools we have to capture
it, might lead us to believe that, either, the level of our mathematical
knowledge is not advanced enough, or, we are using the wrong methods.
Many interesting situations in life (See section \ref{sec:Discussion-of-Other}),
are caused by the uncertainty inherent in them, which we (all researchers
and society) seem to be looking to solve using logic (mathematics).
This paper is meant to illustrate why perhaps dealing with uncertainty,
something twisted, requires an equally twisted approach and perhaps,
solutions might not be obtained using straight techniques that rely
on precision. The dichotomy between logic and randomness is a topic
worth pursuing on many fronts. Our innovation is one possible alternative
methodology, succinctly expressed as, ``Fighting Uncertainty with
Uncertainty''. 

This technique is suitable for the social sciences since the primary
source of uncertainty are the members of the system themselves and
presently, no methods are known to fully determine the outcomes in
such an environment, which perhaps, would require being able to read
the minds of everyone involved and to anticipate their actions continuously.
Admittedly, we are not qualified to recommend whether such an approach
is conducive for the natural sciences, unless perhaps, bounds can
be established on the levels of uncertainty in a system and it is
shown that a better understanding of the system and hence improved
decision making will not alter the outcomes. Barring such a bound
on the level of uncertainty it is advisable to understand the sources
that cause arbitrary outcomes and follow more traditional methods.

The central innovation can be understood as optimizing to get an interval
of satisfactory performance and randomizing over that interval. The
goal is not just to optimize, but to identify a region of acceptable
performance and set the parameters of a probability distribution over
that region, and sample from it to provide an agreeable level of performance.
Agents looking to optimize will randomize over an optimal region,
hence this approach is called \textquotedbl{}Randoptimization\textquotedbl{}.
Due to measurement errors and other uncertainty, we can never be certain
of any optimization we perform, rather it is better to randomize over
acceptable states.

The example we develop in greater detail is the news-vendor inventory
management problem with demand uncertainty. The proposition and theorem
are new results and they depend on existing results which are given
as Lemmas. We briefly discuss areas, where such an approach might
be helpful, with the common prescription, ``Don't Simply Optimize,
Also Randomize; perhaps best described by the term - Randoptimization''.
\end{doublespace}
\begin{enumerate}
\begin{doublespace}
\item News-vendor Inventory Management Problem (Section \ref{sec:News-vendor-Problem}).
\item School Admissions.
\item Journal Submissions.
\item Job Candidate Selection.
\item Stock Picking.
\item Monetary Policy
\end{doublespace}
\end{enumerate}
\begin{doublespace}

\section{\label{sec:News-vendor-Problem}News-vendor Problem}
\end{doublespace}

\begin{doublespace}
The first of our case studies delves into the classical and well studied
problem in the supply chain management literature and is purely about
inventory control. We start with a basic model (Arrow, Harris \& Marschak
\cite{key-11}; Chen and Federgruen \cite{key-12}; Cachon \& Lariviere
\textbf{\cite{key-13})} and introduce the notion of random order
quantity under different levels of uncertainty. (Khouja \cite{key-14};
Qin, Wang, Vakharia, Chen \& Seref \cite{key-15}) provide a comprehensive
review and suggestions for future research. Interesting extensions
include (Gallego \& Moon \cite{key-16}), situations where the distribution
is unknown but the mean and variance are known; (Gallego \& Van Ryzin
\cite{key-17}), changing the price dynamically; (Wu, Li, Wang \&
Cheng \cite{key-18}) stockout costs under a mean variance framework;
(Wilson \& Sorochuk \cite{key-19}), revenues from a secondary stream
related to the product being sold are factored in; (Rubio-Herrero,
Baykal-Gürsoy \& Ja\'{s}kiewicz \cite{key-20}), setting the price
under a mean variance framework. 

(Liyanage \& Shanthikumar \cite{key-20-1}) propose an approach called
operational statistics, which integrates the tasks of parameter estimation
and optimization (maximizing expected profit for the single period
news-vendor model). When no assumptions are made about the form of
the demand distribution, (Chu, Shanthikumar \& Shen \cite{key-20-2})
derive operational statistics that maximize the performance uniformly
for all values of the unknown demand parameters using Bayesian analysis.

To set the stage for introducing our innovation, we first delve into
the basics of inventory theory in Lemma \ref{prop:The-variance-of}
and Lemma \ref{prop:The-optimal-quantity}. The solution, based on
our innovation, is not to find an optimal quantity to order, but to
find an optimal distribution, from which a sample is drawn to give
the order quantity, with the objective of maximizing the expected
profit. The key findings are outlined in Theorem \ref{Theorem:The-profits-in}
and Proposition \ref{prop:If-the-expected}, which are novel, to the
best of our knowledge. 
\end{doublespace}
\begin{doublespace}

\subsection{Notation and Terminology}
\end{doublespace}

\begin{doublespace}
Let $Q$ represent the amount ordered by the retailer before the selling
season starts. Let, the unit cost to manufacture the product be $c$,
the wholesale price at which manufacturer produces and sells to the
retailer be $w$, the salvage value of any unsold product is $s$
per unit, and the stockout cost of unsatisfied demand is $r$ per
unit. The final price at which retailer sells is $p$. Except in the
benchmark model, to avoid unrealistic and trivial cases, we assume
that $0<s<w<p$ , $p>c$ and $0<r$. Retailer faces stochastic demand,
$D$, with Cumulative distribution $F$ and density $f$. When $Q$
is stochastic, it has cumulative distribution $G$ and density $g$.
We assume that $F,G$ are continuous and strictly increasing functions
and $f,g$ are non-negative functions. $Q$ and $D$ have joint distribution
function $F_{Q,D}\left(q,d\right)\Longleftrightarrow F_{X,Y}\left(x,y\right)$
and joint density function $f_{Q,D}\left(q,d\right)\Longleftrightarrow f_{X,Y}\left(x,y\right)$.
Notice that we write $X$ for $Q$ and $Y$ for $D$ to prevent confusions
later when we work with calculus notations. We further assume $Q$
and $D$ are independent giving, $f_{X,Y}\left(x,y\right)=g\left(x\right)f\left(y\right)$
and $F_{X,Y}\left(x,y\right)=G\left(x\right)F\left(y\right)$. We
will relax this assumption of independence and assume that the covariance
between the processes is given by, $\sigma_{XY}$ in subsequent papers
(Kashyap \cite{key-20-3}).
\end{doublespace}
\begin{itemize}
\begin{doublespace}
\item $D\sim F\left[0,u\right]$, $D$ is distributed according to the cumulative
distribution $F$ over the interval $\left[0,u\right]$. 
\item $F,$ is increasing and has full support, which is the non-negative
real line $\left[0,\infty\right]$.
\item $f=F',$ is the continuous density function of $F$.
\item $D\sim U\left[a,b\right]$, if we consider the uniform distribution.
\item $D\sim LN\left[0,u\right]$, if we consider the log normal distribution.
\item $Q\sim G\left[0,v\right]$, $Q$ is distributed according to the cumulative
distribution $G$ over the interval $\left[0,v\right]$. 
\item $G,$ is increasing and has full support, which is the non-negative
real line $\left[0,\infty\right]$.
\item $g=G',$ is the continuous density function of $G$.
\item $Q\sim U\left[u,v\right]$, if we consider the uniform distribution.
\item $Q\sim LN\left[0,v\right]$, if we consider the log normal distribution.
\item $Q,D\sim F_{X,Y}\left(x,y\right)$, $Q$ and $D$ have joint distribution
function $F_{X,Y}\left(x,y\right)=F\left(x\right)G\left(y\right)$
and joint density function $f_{X,Y}\left(x,y\right)=g\left(x\right)f\left(y\right)$
because of independence. When we relax the assumption of independence,
$\text{cov}\left(Q,D\right)=\text{cov}\left(X,Y\right)=\sigma_{XY}$.
\item $\forall$ $a$, $b$ $\in\Re_{+}$, $\quad$$a^{+}=\max\left\{ a,0\right\} $,
$\quad$$a^{-}=\max\left\{ -a,0\right\} $, and$\quad$ $a\bigwedge b=\min\left\{ a,b\right\} $.
\item $\pi_{R}$, the expected profits of the retailer when only Demand,
$D$, is stochastic and $Q$ the quantity ordered is the control variable.
\item $Q^{*}$, the optimal quantity to order when only Demand, $D$, is
stochastic, under the true distribution, $F\left(y\right)$.
\item $\pi_{R}^{*}$, the optimal expected profits of the retailer when
only Demand, $D$, is stochastic and $Q$ the quantity ordered is
the control variable, under the true distribution, $F\left(y\right)$.
\item $\pi_{RS}$, the expected profits of the retailer when Demand, $D$,
is stochastic and $Q$ the quantity ordered is also stochastic.
\item $\tilde{F}\left(y\right),\tilde{f}\left(y\right)$, the estimated
demand distribution and density, which will be different from the
actual distribution, $F\left(y\right)\neq\tilde{F}\left(y\right)\Rightarrow f\left(y\right)\neq\tilde{f}\left(y\right)$,
due to measurement errors and other issues.
\item $\hat{F}\left(y\right),\hat{f}\left(y\right)$, the compound demand
distribution and density, with the assumption that the parameters
of the estimated distribution, $\tilde{F}\left(y\right)$, have distributions
of their own, due to measurement errors and other issues.
\item $\hat{F}_{X,Y}\left(x,y\right)=G\left(x\right)\hat{F}\left(y\right)$,
$\hat{f}_{X,Y}\left(x,y\right)=g\left(x\right)\hat{f}\left(y\right)$,
are the joint distribution function and joint density function of
the compound demand process, $\hat{F}\left(y\right),\hat{f}\left(y\right)$,
and the stochastic quantity ordered. 
\item $\hat{Q^{*}}$, the actual quantity that will be ordered, with the
optimization performed using the estimated distribution, $\tilde{F}\left(y\right)$.
\item $\hat{\pi_{R}^{*}}$, the actual expected profits under the compound
distribution, $\hat{F}\left(y\right)$ when the quantity ordered is
$\hat{Q^{*}}$.
\end{doublespace}
\end{itemize}
\begin{doublespace}

\subsection{Benchmark Model}
\end{doublespace}

\begin{doublespace}
Before we look at stochastic variations in the quantity to be ordered,
we collect results when the retailer optimizes his profit function
to obtain the optimal amount since it offers a comparison point for
the later scenarios. We begin with the simplest possible scenario
by setting $s=r=0$. This assumption implies there is no stockout
costs (no loss of goodwill or direct noticeable cost) if a demand
is not fulfilled by the retailer to both the manufacturer and retailer
and there is no salvage value. We assume that this is a full information
scenario, that is a game in which both parties know about the constraints
and incentives that the other party faces and is looking to optimize
in what environment. The retailer maximizes his profit based on the
following criteria,

\begin{eqnarray*}
\pi_{R} & = & \text{if }\left(D>Q\right),\\
 &  & \quad\left\{ Qp-Qw\right\} \\
 &  & \text{else if }\left(D<Q\right),\\
 &  & \quad\left\{ Dp-Qw\right\} 
\end{eqnarray*}
\[
E\left(\pi_{R}\right)=E_{D}\;\left[\min\left(Q,D\right)p-Qw\right]
\]
\[
\pi_{R}^{*}=\underset{Q}{\max}\quad E_{D}\;\left[\min\left(Q,D\right)p-Qw\right]
\]

\end{doublespace}
\begin{lem}
\begin{doublespace}
\label{prop:The-variance-of}The mean and variance of the profits
are given by,
\[
E\left(\pi_{R}\right)=\left(p-w\right)Q-p\int_{0}^{Q}F\left(t\right)dt
\]
\[
Var\left(\pi_{R}\right)=p^{2}\left[\int_{0}^{Q}2\left(Q-t\right)F\left(t\right)dt-\left\{ \int_{0}^{Q}F\left(t\right)dt\right\} ^{2}\right]
\]
\end{doublespace}
\end{lem}
\begin{proof}
\begin{doublespace}
See Appendix \ref{subsec:Proof-of-Proposition-1}
\end{doublespace}
\end{proof}
\begin{lem}
\begin{doublespace}
\label{prop:The-optimal-quantity}The optimal quantity to order, the
optimal expected profits of the retailer and the variance when the
optimal quantity is ordered are given by,\textbf{
\[
Q^{*}=F^{-1}\left(1-\frac{w}{p}\right)
\]
}
\begin{eqnarray*}
\pi_{R}^{*}= & Q^{*}\left(p-w\right)-p\int_{0}^{Q^{*}}F\left(t\right)dt= & p\int_{0}^{Q^{*}}tf\left(t\right)dt
\end{eqnarray*}
\textbf{
\begin{eqnarray*}
\pi_{R}^{*} & = & p\int_{0}^{Q^{*}}P\left(D>t\right)dt-Q^{*}w
\end{eqnarray*}
}
\begin{eqnarray*}
Var\left(\pi_{R}^{*}\right) & = & p^{2}\left[\int_{0}^{Q^{*}}2\left(Q^{*}-t\right)F\left(t\right)dt-\left\{ \int_{0}^{Q^{*}}F\left(t\right)dt\right\} ^{2}\right]
\end{eqnarray*}
\begin{eqnarray*}
Var\left(\pi_{R}^{*}\right) & = & p^{2}\left[\left(Q^{*}\right)^{2}\left\{ 1-\left(\frac{w}{p}\right)^{2}\right\} -\left(\int_{0}^{Q^{*}}tf\left(t\right)dt\right)^{2}-2Q^{*}\frac{w}{p}\int_{0}^{Q^{*}}tf\left(t\right)dt-2\int_{0}^{Q^{*}}tF\left(t\right)dt\right]
\end{eqnarray*}
\end{doublespace}
\end{lem}
\begin{proof}
\begin{doublespace}
See Appendix \ref{subsec:Proof-of-Proposition}.
\end{doublespace}
\end{proof}
\begin{doublespace}

\subsection{Stochastic Quantity Ordered}
\end{doublespace}

\begin{doublespace}
Any distribution, $\tilde{F}\left(y\right)$, we estimate for the
demand will be different from the true distribution, $F\left(y\right)\neq\tilde{F}\left(y\right)\Rightarrow f\left(y\right)\neq\tilde{f}\left(y\right)$,
due to measurement errors, change in preferences of agents and other
matters that make predictions go awry. This means that the parameters
of the estimated distribution can be further assumed to have distributions
of their own. A reasonable assumption can be that the parameters of
the estimated distribution are normal (or any other suitable one),
making the demand distribution a compound distribution, $\hat{F}\left(y\right)$
with density $\hat{f}\left(y\right)$. Hence the actual quantity we
order, $\hat{Q^{*}}$, could be different from the true optimal quantity,
$\hat{Q^{*}}\neq Q^{*}$, and and hence the profits that actually
accrue, $\hat{\pi_{R}^{*}}$, could be lower than the true optimal
profits, $\hat{\pi_{R}^{*}}\leq\pi_{R}^{*}$.

It is important to keep in mind that, the actual quantity ordered,
$\hat{Q^{*}}$, is the result of an optimization performed using the
estimated distribution, $\tilde{F}\left(y\right)$; but the profits
that will accrue, $\hat{\pi_{R}^{*}}$, when we order this quantity,
$\hat{Q^{*}}$, is the result obtained, when operating in an environment
where the compound distribution, $\hat{F}\left(y\right)$, holds influence.
To counter this measurement error, we introduce randomness in the
number of items ordered. This means, $Q$ representing the amount
ordered by the retailer, before the selling season starts, will be
stochastic, with cumulative distribution $G$ and density $g$. $Q$
and $D$ have joint distribution function $\hat{F}_{X,Y}\left(x,y\right)=G\left(x\right)\hat{F}\left(y\right)$
and joint density function $\hat{f}_{X,Y}\left(x,y\right)=g\left(x\right)\hat{f}\left(y\right)$
because of independence.
\end{doublespace}
\begin{thm}
\begin{doublespace}
\label{Theorem:The-profits-in}When the quantity ordered is stochastic:
the expected profits, $E\left[\pi_{RS}\right]$, will be greater than,
or equal to, the expected profits of the benchmark model, $\hat{\pi_{R}^{*}}$,
such that the condition, $E\left[\pi_{RS}\right]\geq\hat{\pi_{R}^{*}}$
is satisfied, ensuring that this approach will yield better outcomes,
as long as the following restriction holds on the parameters. 
\begin{eqnarray*}
E\left[Q\right]\left[1-\frac{w}{p}\right]+E\left[D\right]-E\left[\max\left(Q,D\right)\right]\quad & \geq & \quad\int_{0}^{\hat{Q^{*}}}t\hat{f}\left(t\right)dt\\
\left\{ \text{Here,}\quad\hat{Q^{*}}=\tilde{F}^{-1}\left(1-\frac{w}{p}\right)\right\} 
\end{eqnarray*}
The expected profits, $E\left[\pi_{RS}\right]$, are with resect to
the joint distribution, $\hat{F}_{X,Y}\left(x,y\right)=G\left(x\right)\hat{F}\left(y\right)$,
where demand is governed by the compound distribution, $\hat{F}\left(y\right)$.
The expected profits of the benchmark model, $\hat{\pi_{R}^{*}}$,
are with respect to the compound distribution, $\hat{F}\left(y\right)$,
when the actual quantity ordered, $\hat{Q^{*}}$, is obtained by optimizing
under the estimated distribution, $\tilde{F}\left(y\right)$.
\end{doublespace}
\end{thm}
\begin{proof}
\begin{doublespace}
See Appendix \ref{subsec:Proof-of-Theorem}.
\end{doublespace}
\end{proof}
\begin{doublespace}
Instead of optimizing the quantity ordered with respect to the compound
distribution of demand with measurement errors, randomizing is better
due to the errors in measurement errors and so on, ad infinitum (Taleb
\cite{key-21}). All of these errors are extremely hard to forecast,
considering that a big unknown is how many recursive levels of errors
we might encounter, or, what is a satisfactory level to stop. 

A central question is: how will buyers respond to different inventory
management techniques. (DeGraba \cite{key-22}) investigates tactics
employed by sellers to induce excess demand, (leading to higher prices),
by creating artificial shortages. The message from the usage of such
methods is that we can randomize across a subset of states for the
quantity ordered, or, pick randomly from a few of the most likely
possibilities, (or the desired set, that need not be the most likely),
for the quantity ordered. This can also be understood as optimizing
to get an interval of acceptable performance and randomizing over
that interval; given our questionable ability to get the best outcome,
we go with a mix of the better outcomes. Hence, the prescription,
``Don't Simply Optimize, Also Randomize; perhaps best described by
the term - Randoptimization''. Since we control the choice of the
distribution for $Q,$ we note a nice property in the following proposition.
\end{doublespace}
\begin{prop}
\begin{doublespace}
\label{prop:If-the-expected}The parameters of the distribution for
the stochastic quantity ordered can be got by solving the below expression,
if its expected value is set such that $E\left[Q\right]=\hat{Q^{*}}$.
\[
\int_{0}^{\infty}xg\left(x\right)\hat{F}\left(x\right)dx+\int_{0}^{\infty}y\hat{f}\left(y\right)G\left(y\right)dy\leq\hat{Q^{*}}\hat{F}\left(\hat{Q}^{*}\right)+\int_{\hat{Q}^{*}}^{\infty}t\hat{f}\left(t\right)dt
\]
\end{doublespace}
\end{prop}
\begin{proof}
\begin{doublespace}
See Appendix \ref{subsec:Proof-of-Proposition-2}.

The result in proposition \ref{prop:If-the-expected} can be viewed
as an extra set of constraints on the parameters of the distribution,
from which we have to choose the quantity to order. This simplification
has an intuitive appeal that the expected value of the stochastic
quantity ordered is equal to the optimal estimation of the quantity
to order, in the presence of measurement errors. Depending on the
distribution for the demand and the distribution from which we sample
the quantity to order, this criteria can turn out to be highly restrictive,
in which case the more relaxed results from theorem \ref{Theorem:The-profits-in}
need to be used instead.
\end{doublespace}
\end{proof}
\begin{doublespace}

\section{\label{sec:Discussion-of-Other}Discussion of Other Applications}
\end{doublespace}

\begin{doublespace}
We provide a brief synopsis of how this concept can be useful in other
areas, with a complete development of each sub-topic planned in subsequent
papers.
\end{doublespace}
\begin{doublespace}

\subsection{School Admissions}
\end{doublespace}

\begin{doublespace}
University rankings and the rush to get in to the best schools, as
determined by these rankings, is well acknowledged and studied. The
prestige associated with having attended a top school lingers on for
a lifetime and beyond. The admission process and how valid or effective
it is in determining determining academic success or corporate productivity
are debatable at best: (Fishman and Pasanella \cite{key-23}; Deckro
\& Woundenberg \cite{key-24}; Willingham \cite{key-25}; Sorensen
\cite{key-26}; Sandow, Jones, Peek, Courts \& Watson \cite{key-27}). 

While aspirations for higher rankings and the quest for excellence
are worthy pursuits, what rankings have created are false targets
and misleading perceptions. The goal should not simply be to get into
the best possible school, but to get into a good school and get the
best possible education. No doubt, admission to an elite university
can help obtain a great education, but the difference in the quality
of teaching at most good universities is not that significant. Without
going into too many details, we mention that people go to great lengths
simply to secure admission into the top colleges. The current set
up of rankings and the best ranked institutions attracting the best
minds and resources, become self fulfilling prophecies that segment
society, which is one of the primary outcomes that education seeks
to eradicate.
\end{doublespace}
\begin{conjecture}
\begin{doublespace}
When making offers of admission to schools, it can be made publicly
known that the selection will be done at random from a pool of candidates
who meet a certain number of basic qualifying criteria. Different
universities could even collaborate to select students from the overall
pool in this way, after considering applicable constraints and preferences.
This will ensure that the rankings will change as the world changes
and the focus will not be on getting into the best university, but
to get into a good university and getting the best possible education.
\end{doublespace}
\end{conjecture}
\begin{doublespace}
For simplicity, the discussion above does not distinguish between
admission to universities for undergraduate education and for admission
into elementary and high schools, since the same solution can apply
to all these cases with minor adjustments. Postgraduate education
can also use these methods, but a greater level of restrictions needs
to be placed, to match the preferences of candidates regarding the
desired specialized field of study and the strengths of the higher-level
educational institution. Of course, certain exceptions can always
be made, wherein, really deserving candidates and truly exceptional
students can be granted admission, without subjecting them to the
above procedure.
\end{doublespace}
\begin{doublespace}

\subsection{Job Candidate Selection}
\end{doublespace}

\begin{doublespace}
Admissions to schools is still a much more open and fair process when
compared to the job market, where cronyism, connections and catchy
phrases written on resumes, which manage to get the attention of hiring
managers, seem to be the key. A random process, fitting with the central
theme of our discussion, could select the successful candidate from
the pool of applicants meeting certain benchmark criteria. Of course,
human resource personnel, recruiters and hiring managers, who pride
themselves on finding the perfect fit for their organization, would
voice serious concerns. We are not looking to malign or undermine
that effort. But would it really make such a big difference how perfect
the candidate as long as he is not a drug addict or a violent criminal
and he can learn the skills required to perform the job once hired?
Human resources can then focus on training programs and other means
of actually enhancing productivity.
\end{doublespace}
\begin{doublespace}

\subsection{Journal Submissions}
\end{doublespace}

\begin{doublespace}
A topic that would perhaps determine the fate of this paper and others
we plan to write developing these topics is about the current selection
process at the top journals. But when papers such as these are being
written, {[}the titles of these papers are quite sufficient to understand
the nature of the issues they address: (Zivney \& Bertin \cite{key-28}),
``Publish or perish: What the competition is really doing''; (Seglen
\cite{key-29}), ``Why the impact factor of journals should not be
used for evaluating research''; (Choi \cite{key-30}), ``How to
publish in top journals''; (Lawrence \cite{key-31}), ``The politics
of publication''; (De Rond \& Miller \cite{key-32}), ``Publish
or Perish: Bane or Boon of Academic Life?''{]}, it is a warning sign
that maximum efforts are not being spent on creating the best work,
but a lot of work is spent in getting published at the top outlets. 

With no offense to any of the editors and reviewers who spend countless
hours finding the most suitable papers for their journals. Would it
not be easier and better to randomly select good papers from the overall
pool of submissions and coach the authors so that the papers can develop
to become excellent papers? Different journals could even collaborate
to select papers from the submission pool in this way. This will perhaps
also ensure that papers will not just be judged merely by where they
are published but a deeper evaluation will be done by future authors
to determine what previous work is helpful for them. It can be argued
that such a thorough assessment will lead to better fulfillment of
one of the fundamental goals of research and the pursuit of knowledge,
which is to promote better decision making. No doubt, the problem
here is not that severe since the best works do bubble up to the surface
over time and researchers are careful in selecting the works they
deem beneficial for further studying or enhancing. We merely include
this example as another instance where introducing uncertainty in
the decision process can yield better outcomes over time.
\end{doublespace}
\begin{doublespace}

\subsection{Stock Picking}
\end{doublespace}

\begin{doublespace}
No discussion on uncertainty in the social sciences is complete without
talking about the financial markets. The amount of effort devoted
to security selection and chasing higher returns is colossal. A simple
mechanism could pick a pool of securities meeting certain criteria
and randomly allocate wealth to a smaller percentage of the overall
pool. This would acknowledge the randomness (Taleb \cite{key-33})
inherent in the outcomes of star portfolio managers and re-channel
efforts from financial management to other more worthy endeavors.
\end{doublespace}
\begin{doublespace}

\subsection{Capitol Hill Baby Sitting Crisis and Monetary Policy Management}
\end{doublespace}

\begin{doublespace}
Lastly, we consider, perhaps, one of the most elegant, simple and
beautiful examples of uncertainty in social systems. The (Sweeney
\& Sweeney \cite{key-34}) anecdote about the Capitol Hill baby-sitting
crisis, exposits the mechanics of inflation, setting interest rates
and monetary policies required to police the optimum amount of money.
The creation of a monetary crisis in a small, simple environment of
good-hearted people, expounds that even with near ideal conditions,
things can become messy; then in a large labyrinthine atmosphere,
disaster could be brewing without getting noticed and can strike without
much premonition. 

The primary emphasis should be on the need to keep complexity at bay
and establishing an ambience, where repeated games can be played with
public transparency, so that guileful practices can be curtailed.
A solution based on the central theme of our paper, could be to issue
lotteries and award scrips (fiat money that could be exchanged in
lieu of baby sitting hours) to the residents when there is a recession
(too many baby sitters and not enough baby sitting jobs or money to
be made) and expire a certain amount of scrips either periodically
or even better on randomly selected dates when there are inflationary
pressures (too much money or scrip chasing too few good baby sitters).
The simplicity of our solution is a stark contrast to the elegant
yet complex approach, filled with assumptions, outlined in (Hens,
Schenk-Hoppé \& Vogt \cite{key-35}). A detailed comparison and related
treatments are postponed to another time.
\end{doublespace}
\begin{doublespace}

\section{Limitations and Scope for Further Research}
\end{doublespace}

\begin{doublespace}
In the interest of brevity, we have only provided initial results
for the inventory management case and briefly discussed other examples. 
\end{doublespace}
\begin{itemize}
\begin{doublespace}
\item The inventory management problem can be extended to consider example
distributions and estimation of parameters that can satisfy the inequalities
derived (Brunk \cite{key-36}; Parzen \cite{key-37}; Nickalls \cite{key-38};
Jaakkola \& Jordan \cite{key-39}; Gibbs \& Su \cite{key-40}; Weisstein
\cite{key-41}). The hunt for distributions for the quantity to order
can be quite laborious and could involve both discrete and continuous
varieties. When no solutions exist or when solutions are not easily
obtained, a limited number of nearly optimal states need to be identified
and their profit potential ascertained, keeping in mind the extent
of maximum losses that could be incurred.
\item The main result in Theorem \ref{Theorem:The-profits-in} is simplistic
in nature, but is mostly meant to illustrate that our approach, where
the quantity ordered is stochastic, holds promise, and further research
in this direction can be rewarding. 
\item Also for simplicity, we have assumed independence between the demand
process and the stochastic quantity ordered. The number of feasible
solutions in terms of the distribution parameters for the stochastic
quantity might be limited or not be easily obtained in this case.
We can assume covariance between these two random variables to obtain
more valid solutions. We might expect that when the two variables
are positively correlated, the output, in terms of profits, might
be more beneficial. These results, need to be established theoretically
and also empirically verified, to provide guidance on the choice of
distributions for the quantity to be ordered and the extent of correlation
with the demand process.
\item A deeper analysis could yield a useful set of comparisons, that establish
bounds on the difference in profits under the true demand distribution,
the estimated demand distribution (a good example of which is the
empirical distribution function) and the scenario when the quantity
ordered is stochastic. Such bounds under both finite sample situations
and asymptotically, could shed light on the criteria when the stochastic
quantity would be a beneficial choice.
\item A realistic assumption would be that extreme demand situations or
outliers would be underrepresented in a historical time series of
observed demand. It would be interesting to understand the behavior
of the profits when the variance of the demand distribution varies
and when outliers occur. The choice of the distribution for the stochastic
quantity to order should aim to push the system into regions where,
despite any errors in estimation, profits are not significantly affected
and if possible even benefit from the fluctuations of the demand process.
\item Different utility functions or mean variance based optimization techniques
can also be introduced into this framework.
\item We are working on creating subsequent papers to significantly enhance
the results for inventory management, to develop theoretical platforms
and test empirically the other case studies mentioned here.
\end{doublespace}
\end{itemize}
\begin{doublespace}
The dynamic nature of any social-science system, where changes can
be observed and decisions can be made by participants to influence
the system, means that the limited predictive ability of any awareness
will necessitate periodic reviews and the prescription of corrective
programs. It would be interesting to see how participants modify their
actions once randomness is introduced in the decision making process.
Where there is uncertainty, there will be unintended consequences,
which might be welcome or hazardous (Kashyap \cite{key-41-1}). Of
interest would be to see whether such an approach will reduce the
sense of entitlement prevalent in society or whether it would lead
to greater complacency.
\end{doublespace}
\begin{doublespace}

\section{Conclusion}
\end{doublespace}

\begin{doublespace}
We have developed a framework using the inventory management example,
to illustrate how we can overcome uncertainty with uncertainty. Using
randomness in our choices and in what we control, and hence in the
decision making process, could potentially offset the uncertainty
inherent in the environment and yield better outcomes. Such an approach,
while seemingly absurd, is the ideal medicine for the even greater
absurdity in the decision making process, that has become prevalent
in today's society. It takes care of the issues that crop up due to
the limitations in our understanding of complex systems and the widely
acknowledged problem that most of our measurements are highly prone
to errors.

The central innovation can be understood as optimizing to get an interval
of adequate performance and randomizing over that interval. The goal
is not just to optimize, but to identify a region of optimal performance
and set the parameters of a probability distribution over that region
and sample from it to provide a satisfactory level of performance.
Agents looking to optimize will randomize over an optimal region,
hence this approach is called \textquotedbl{}Randoptimization\textquotedbl{}.
Due to measurement errors and other uncertainty, we can never be certain
of any optimization we perform, rather it is better to randomize over
acceptable states.

We have discussed areas where such an approach might be suited, with
the common prescription, ``Don't Simply Optimize, Also Randomize;
perhaps best described by the term - Randoptimization''.
\end{doublespace}
\begin{enumerate}
\begin{doublespace}
\item News-vendor Inventory Management Problem
\item School Admissions
\item Journal Submissions
\item Job Candidate Selection
\item Stock Picking
\item Monetary Policy
\end{doublespace}
\end{enumerate}

\begin{doublespace}

\section{Appendix}
\end{doublespace}
\begin{doublespace}

\subsection{\label{subsec:Proof-of-Proposition-1}Proof of Lemma \ref{prop:The-variance-of}}
\end{doublespace}
\begin{proof}
\begin{doublespace}
The proof is standard and well known; but we provide it because of
its central importance to this paper and for completeness. Consider
the profit function of the benchmark model,
\[
\pi_{R}=\min\left(Q,D\right)p-Qw
\]
\[
\pi_{R}=\left\{ Q-\left(Q-D\right)^{+}\right\} p-Qw
\]
\begin{eqnarray*}
\left\{ \vphantom{\int_{0}^{Q}}\because\;\min\left(Q,D\right)\right.=\left\{ Q-\left(Q-D\right)^{+}\right\}  & ; & \text{If }\left(Q<D\right);\left\{ Q-\left(Q-D\right)^{+}\right\} =Q-0=\min\left(Q,D\right)\\
\text{If }\left(Q\geq D\right);\left\{ Q-\left(Q-D\right)^{+}\right\} =Q-\left(Q-D\right)=D=\min\left(Q,D\right) &  & \left.\vphantom{\int_{0}^{Q}}\right\} 
\end{eqnarray*}
\[
\pi_{R}=\left(p-w\right)Q-p\left(Q-D\right)^{+}
\]
\[
E\left(\pi_{R}\right)=E_{D}\left[\left(p-w\right)Q-p\left(Q-D\right)^{+}\right]
\]
\[
E\left(\pi_{R}\right)=\left(p-w\right)Q-p\int_{0}^{Q}F\left(t\right)dt
\]
\begin{eqnarray*}
\left\{ \vphantom{\int_{0}^{Q}}\because\;E_{D}\left[p\left(Q-D\right)^{+}\right]\right. & = & p\int_{0}^{Q}\left(Q-t\right)f\left(t\right)dt=pQ\int_{0}^{Q}f\left(t\right)dt-p\int_{0}^{Q}tf\left(t\right)dt\quad;\\
 &  & \text{Integrating by Parts, using, }\int u\,dv=uv-\int v\,du\,,\\
\quad p\int_{0}^{Q}tf\left(t\right)dt & = & pQ\int_{0}^{Q}f\left(t\right)dt-p\int_{0}^{Q}F\left(t\right)dt\left.\quad;\;pQ\int_{0}^{Q}f\left(t\right)dt=pQF\left(Q\right)\vphantom{\int_{0}^{Q}}\right\} 
\end{eqnarray*}
\[
Var\left(\pi_{R}\right)=E\left[\pi_{R}-E\left(\pi_{R}\right)\right]^{2}=E\left[\left(\pi_{R}\right)^{2}\right]-\left[E\left(\pi_{R}\right)\right]^{2}
\]
\begin{eqnarray*}
Var\left(\pi_{R}\right) & = & E_{D}\left[\left(p-w\right)^{2}Q^{2}+p^{2}\left\{ \left(Q-D\right)^{+}\right\} ^{2}-2\left\{ \left(p-w\right)pQ\left(Q-D\right)^{+}\right\} \right]\\
 &  & -\left[\left(p-w\right)^{2}Q^{2}+p^{2}\left\{ \int_{0}^{Q}F\left(t\right)dt\right\} ^{2}-2\left\{ \left(p-w\right)pQ\int_{0}^{Q}F\left(t\right)dt\right\} \right]
\end{eqnarray*}
\begin{eqnarray*}
Var\left(\pi_{R}\right) & = & E_{D}\left[p^{2}\left\{ \left(Q-D\right)^{+}\right\} ^{2}\right]-p^{2}\left\{ \int_{0}^{Q}F\left(t\right)dt\right\} ^{2}
\end{eqnarray*}
\begin{eqnarray*}
Var\left(\pi_{R}\right) & = & p^{2}\left[\int_{0}^{Q}2\left(Q-t\right)F\left(t\right)dt-\left\{ \int_{0}^{Q}F\left(t\right)dt\right\} ^{2}\right]
\end{eqnarray*}
\begin{eqnarray*}
\left\{ \vphantom{\int_{0}^{Q}}\because\;E_{D}\left[\left\{ \left(Q-D\right)^{+}\right\} ^{2}\right]\right. & = & \int_{0}^{Q}\left(Q-t\right)^{2}f\left(t\right)dt\quad,\;\text{Integrate this by Parts}\\
 & = & \left.\left|\left(Q-t\right)^{2}F\left(t\right)\right|_{0}^{Q}-\int_{0}^{Q}2\left(Q-t\right)\left(-1\right)F\left(t\right)dt\quad\right\} 
\end{eqnarray*}
\end{doublespace}
\end{proof}
\begin{doublespace}

\subsection{\label{subsec:Proof-of-Proposition}Proof of Lemma \ref{prop:The-optimal-quantity}}
\end{doublespace}
\begin{proof}
\begin{doublespace}
The proof is standard and well known; but we provide it because of
its central importance to this paper and for completeness. Consider
maximizing the profit function of the benchmark model,\textbf{ 
\begin{eqnarray*}
\pi_{R}^{*} & = & \underset{Q}{\max}\quad E_{D}\left[\min\left(Q,D\right)p-Qw\right]\\
 & = & \max\:\left[\int_{0}^{Q}ptf\left(t\right)dt+\int_{Q}^{\infty}Qpf\left(t\right)dt-Qw\right]\\
\left\{ \phantom{\left[\int_{Q}^{P}\right]}\text{we write this instead of}\right. &  & \max\left.\left[\int_{0}^{Q}pdf\left(d\right)dd+\int_{Q}^{\infty}Qpf\left(d\right)dd-Qw\right]\right\} 
\end{eqnarray*}
}First Order Conditions using Leibniz Integral Rule gives,
\begin{eqnarray*}
\frac{\partial\left[\int_{0}^{Q}ptf\left(t\right)dt+\int_{Q}^{\infty}Qpf\left(t\right)dt-Qw\right]}{\partial Q} & = & 0\\
pQf\left(Q\right)+\int_{Q}^{\infty}pf\left(t\right)dt-Qpf\left(Q\right)-w & = & 0\\
\Rightarrow\left\{ F\left(\infty\right)-F\left(Q\right)\right\}  & = & \frac{w}{p}\\
F\left(Q\right) & = & 1-\frac{w}{p}
\end{eqnarray*}
This gives the optimal quantity to order for the retailer as,\textbf{
\[
Q^{*}=F^{-1}\left(1-\frac{w}{p}\right)
\]
}The optimal expected profits of the retailer would then be
\[
\pi_{R}^{*}=E_{D}\left[\min\left(Q^{*},D\right)p-Q^{*}w\right]
\]
\[
\pi_{R}^{*}=p\int_{0}^{Q^{*}}tf\left(t\right)dt+\int_{Q^{*}}^{\infty}Q^{*}pf\left(t\right)dt-Q^{*}w
\]
\[
\pi_{R}^{*}=p\int_{0}^{Q^{*}}tf\left(t\right)dt+Q^{*}p\left\{ 1-\int_{0}^{Q^{*}}f\left(t\right)dt\right\} -Q^{*}w
\]
\[
\pi_{R}^{*}=\int_{0}^{Q^{*}}ptf\left(t\right)dt+Q^{*}p\left\{ 1-F\left(Q^{*}\right)\right\} -Q^{*}w
\]
\begin{eqnarray*}
\pi_{R}^{*} & = & \int_{0}^{Q^{*}}ptf\left(t\right)dt+Q^{*}w-Q^{*}w\\
\pi_{R}^{*} & = & p\int_{0}^{Q^{*}}tf\left(t\right)dt
\end{eqnarray*}
This can also be written as,
\[
\pi_{R}^{*}=p\int_{0}^{Q^{*}}tf\left(t\right)dt-Q^{*}p\int_{0}^{Q^{*}}f\left(t\right)dt+Q^{*}p-Q^{*}w
\]
\[
\pi_{R}^{*}=p\int_{0}^{Q^{*}}tf\left(t\right)dt-Q^{*}pF\left(Q^{*}\right)+Q^{*}p-Q^{*}w
\]
\textbf{
\begin{eqnarray*}
\pi_{R}^{*} & = & Q^{*}\left(p-w\right)-p\int_{0}^{Q^{*}}F\left(t\right)dt
\end{eqnarray*}
}
\begin{eqnarray*}
\left\{ \because\quad\int_{0}^{Q^{*}}tf\left(t\right)dt\right. & = & Q^{*}F\left(Q^{*}\right)-\int_{0}^{Q^{*}}F\left(t\right)dt\left.\quad,\;\text{Using Integration by parts}\vphantom{\int_{0}^{Q^{*}}}\right\} 
\end{eqnarray*}
\textbf{
\begin{eqnarray*}
\pi_{R}^{*} & = & p\int_{0}^{Q^{*}}\left\{ 1-F\left(t\right)\right\} dt-Q^{*}w
\end{eqnarray*}
\begin{eqnarray*}
\pi_{R}^{*} & = & p\int_{0}^{Q^{*}}\left\{ 1-P\left(D\leq t\right)\right\} dt-Q^{*}w
\end{eqnarray*}
\begin{eqnarray*}
\pi_{R}^{*} & = & p\int_{0}^{Q^{*}}P\left(D>t\right)dt-Q^{*}w
\end{eqnarray*}
}

The variance of the profits when the optimal quantity is ordered is
given by,
\begin{eqnarray*}
Var\left(\pi_{R}^{*}\right) & = & p^{2}\left[\int_{0}^{Q^{*}}2\left(Q^{*}-t\right)F\left(t\right)dt-\left\{ \int_{0}^{Q^{*}}F\left(t\right)dt\right\} ^{2}\right]
\end{eqnarray*}
\begin{eqnarray*}
 & = & p^{2}\left[2Q^{*}\int_{0}^{Q^{*}}F\left(t\right)dt-2\int_{0}^{Q^{*}}tF\left(t\right)dt-\left\{ \int_{0}^{Q^{*}}F\left(t\right)dt\right\} ^{2}\right]
\end{eqnarray*}
\begin{eqnarray*}
 & = & p^{2}\left[\int_{0}^{Q^{*}}F\left(t\right)dt\left\{ 2Q^{*}-\left(\int_{0}^{Q^{*}}F\left(t\right)dt\right)\right\} -2\int_{0}^{Q^{*}}tF\left(t\right)dt\right]
\end{eqnarray*}
\begin{eqnarray*}
 & = & p^{2}\left[\left\{ Q^{*}F\left(Q^{*}\right)-\int_{0}^{Q^{*}}tf\left(t\right)dt\right\} \left\{ 2Q^{*}-\left(Q^{*}F\left(Q^{*}\right)-\int_{0}^{Q^{*}}tf\left(t\right)dt\right)\right\} -2\int_{0}^{Q^{*}}tF\left(t\right)dt\right]
\end{eqnarray*}
\begin{eqnarray*}
 & = & p^{2}\left[\left\{ Q^{*}\left\{ 1-\frac{w}{p}\right\} -\int_{0}^{Q^{*}}tf\left(t\right)dt\right\} \left\{ Q^{*}\left\{ 1+\frac{w}{p}\right\} +\int_{0}^{Q^{*}}tf\left(t\right)dt\right\} -2\int_{0}^{Q^{*}}tF\left(t\right)dt\right]
\end{eqnarray*}
\begin{eqnarray*}
Var\left(\pi_{R}^{*}\right) & = & p^{2}\left[\left(Q^{*}\right)^{2}\left\{ 1-\left(\frac{w}{p}\right)^{2}\right\} -\left(\int_{0}^{Q^{*}}tf\left(t\right)dt\right)^{2}-2Q^{*}\frac{w}{p}\int_{0}^{Q^{*}}tf\left(t\right)dt-2\int_{0}^{Q^{*}}tF\left(t\right)dt\right]
\end{eqnarray*}
\end{doublespace}
\end{proof}
\begin{doublespace}

\subsection{\label{subsec:Proof-of-Theorem}Proof of Theorem \ref{Theorem:The-profits-in}}
\end{doublespace}
\begin{proof}
\begin{doublespace}
Consider the expected profit function in the stochastic quantity ordered
case,\textbf{
\begin{eqnarray*}
E\left[\pi_{RS}\right] & = & E_{Q,D}\left[\min\left(Q,D\right)p-Qw\right]\\
 & = & p\int_{0}^{\infty}\int_{0}^{\infty}\min\left(x,y\right)\hat{f}_{X,Y}\left(x,y\right)dxdy-E\left[Q\right]w
\end{eqnarray*}
}
\begin{eqnarray*}
 & = & p\int_{0}^{\infty}\int_{0}^{y}xg\left(x\right)\hat{f}\left(y\right)dxdy+p\int_{0}^{\infty}\int_{y}^{\infty}yg\left(x\right)\hat{f}\left(y\right)dxdy-E\left[Q\right]w
\end{eqnarray*}
\begin{eqnarray*}
 & = & p\int_{0}^{\infty}\int_{x}^{\infty}xg\left(x\right)\hat{f}\left(y\right)dydx+p\int_{0}^{\infty}\int_{y}^{\infty}yg\left(x\right)\hat{f}\left(y\right)dxdy-E\left[Q\right]w
\end{eqnarray*}
\begin{eqnarray*}
 & = & p\int_{0}^{\infty}xg\left(x\right)\left\{ \int_{x}^{\infty}\hat{f}\left(y\right)dy\right\} dx+p\int_{0}^{\infty}y\hat{f}\left(y\right)\left\{ \int_{y}^{\infty}g\left(x\right)dx\right\} dy-E\left[Q\right]w
\end{eqnarray*}
\begin{eqnarray*}
 & = & p\int_{0}^{\infty}xg\left(x\right)\left\{ 1-\int_{0}^{x}\hat{f}\left(y\right)dy\right\} dx+p\int_{0}^{\infty}y\hat{f}\left(y\right)\left\{ 1-\int_{0}^{y}g\left(x\right)dx\right\} dy-E\left[Q\right]w
\end{eqnarray*}
\begin{eqnarray*}
 & = & p\int_{0}^{\infty}xg\left(x\right)\left\{ 1-\hat{F}\left(x\right)\right\} dx+p\int_{0}^{\infty}y\hat{f}\left(y\right)\left\{ 1-G\left(y\right)\right\} dy-E\left[Q\right]w
\end{eqnarray*}
\begin{eqnarray*}
 & = & p\int_{0}^{\infty}xg\left(x\right)dx-p\int_{0}^{\infty}xg\left(x\right)\hat{F}\left(x\right)dx+p\int_{0}^{\infty}y\hat{f}\left(y\right)dy-p\int_{0}^{\infty}y\hat{f}\left(y\right)G\left(y\right)dy-E\left[Q\right]w
\end{eqnarray*}
\begin{eqnarray*}
 & = & p\int_{0}^{\infty}xg\left(x\right)dx+p\int_{0}^{\infty}y\hat{f}\left(y\right)dy-p\int_{0}^{\infty}xg\left(x\right)\hat{F}\left(x\right)dx-p\int_{0}^{\infty}y\hat{f}\left(y\right)G\left(y\right)dy-E\left[Q\right]w
\end{eqnarray*}
We note that,

\[
E\left[\max\left(X,Y\right)\right]=\int_{0}^{\infty}\int_{0}^{\infty}\max\left(x,y\right)\hat{f}_{X,Y}\left(x,y\right)dxdy
\]
\[
=\int_{0}^{\infty}\int_{0}^{x}xg\left(x\right)\hat{f}\left(y\right)dydx+\int_{0}^{\infty}\int_{0}^{y}yg\left(x\right)\hat{f}\left(y\right)dxdy
\]
\[
=\int_{0}^{\infty}xg\left(x\right)\left\{ \int_{0}^{x}\hat{f}\left(y\right)dy\right\} dx+\int_{0}^{\infty}y\hat{f}\left(y\right)\left\{ \int_{0}^{y}g\left(x\right)dx\right\} dy
\]
\[
=\int_{0}^{\infty}xg\left(x\right)\hat{F}\left(x\right)dx+\int_{0}^{\infty}y\hat{f}\left(y\right)G\left(y\right)dy
\]
This gives,
\[
E\left[\pi_{RS}\right]=pE\left[Q\right]+pE\left[D\right]-pE\left[\max\left(Q,D\right)\right]-E\left[Q\right]w
\]
Alternately, noting that the original specification of the inventory
problem has a $\min\left(Q,D\right)$ we can get the above by the
following relationship between the minimum and maximum,
\[
\min\left({X,Y}\right)+\max\left({X,Y}\right)=X+Y
\]
For better outcomes under the stochastic quantity ordered scenario,
we must have, 
\[
E\left[\pi_{RS}\right]\geq\hat{\pi_{R}^{*}}
\]
\[
pE\left[Q\right]+pE\left[D\right]-pE\left[\max\left(Q,D\right)\right]-E\left[Q\right]w\quad\geq\quad p\int_{0}^{\hat{Q^{*}}}t\hat{f}\left(t\right)dt
\]
\[
E\left[Q\right]\left\{ 1-\frac{w}{p}\right\} +E\left[D\right]-E\left[\max\left(Q,D\right)\right]\quad\geq\quad\int_{0}^{\hat{Q^{*}}}t\hat{f}\left(t\right)dt
\]
\end{doublespace}
\end{proof}
\begin{doublespace}

\subsection{\label{subsec:Proof-of-Proposition-2}Proof of Proposition \ref{prop:If-the-expected}}
\end{doublespace}
\begin{proof}
\begin{doublespace}
From the criteria in Theorem \ref{Theorem:The-profits-in}, the result
is immediate
\[
E\left[\pi_{RS}\right]\geq\hat{\pi_{R}^{*}}
\]
 
\[
E_{Q,D}\left[\min\left(Q,D\right)p-Qw\right]\geq E_{D}\left[\min\left(\hat{Q^{*}},D\right)p-\hat{Q^{*}}w\right]
\]
\[
E_{Q,D}\left[Q+D-\max\left(Q,D\right)p-Qw\right]\geq E_{D}\left[\hat{Q^{*}}+D-\max\left(\hat{Q^{*}},D\right)p-\hat{Q^{*}}w\right]
\]
\[
E_{Q,D}\left[\max\left(Q,D\right)\right]\leq E_{D}\left[\max\left(\hat{Q^{*}},D\right)\right]
\]
\[
\int_{0}^{\infty}xg\left(x\right)\hat{F}\left(x\right)dx+\int_{0}^{\infty}y\hat{f}\left(y\right)G\left(y\right)dy\leq\int_{0}^{\hat{Q^{*}}}\hat{Q^{*}}\hat{f}\left(t\right)dt+\int_{\hat{Q^{*}}}^{\infty}t\hat{f}\left(t\right)dt
\]
\[
\int_{0}^{\infty}xg\left(x\right)\hat{F}\left(x\right)dx+\int_{0}^{\infty}y\hat{f}\left(y\right)G\left(y\right)dy\leq\hat{Q^{*}}\hat{F}\left(\hat{Q^{*}}\right)+\int_{\hat{Q^{*}}}^{\infty}t\hat{f}\left(t\right)dt
\]
\end{doublespace}
\end{proof}

\end{document}